\documentclass[11pt, a4paper, twoside]{amsart}
\usepackage[centering, totalwidth = 370pt, totalheight = 630pt]{geometry}
\usepackage{amssymb, amsmath, amsthm, enumerate, microtype, stmaryrd, url, mathrsfs}
\usepackage[nobreak]{cite}
\usepackage[latin1]{inputenc}
\usepackage[dvips, arrow, matrix, tips, curve]{xy}
\usepackage{mathpartir}
\SelectTips{cm}{10}

\newcommand{\cat}[1]{\mathbf{#1}}

\newcommand{\op}{\mathrm{op}}

\newcommand{\thg}{{\mathord{\text{--}}}}

\renewcommand{\ll}{\gamma}
\newcommand{\rr}{\rho}
\newcommand{\lbl}{\ell}

\newcommand{\abs}[1]{{\left|{#1}\right|}}

\newcommand{\defeq}{\mathrel{\mathop:}=}

\newcommand{\cd}[2][]{\vcenter{\hbox{\xymatrix#1{#2}}}}

\newcommand{\A}{{\mathscr A}}

\newcommand{\E}{{\mathscr E}}

\newcommand{\xtor}[1]{\cdl[@1]{{} \ar[r]|-{\object@{|}}^{#1} & {}}}

\makeatletter

\def\hookleftarrowfill@{\arrowfill@\leftarrow\relbar{\relbar\joinrel\rhook}}
\def\twoheadleftarrowfill@{\arrowfill@\twoheadleftarrow\relbar\relbar}
\def\leftbararrowfill@{\arrowdoublefill@{\leftarrow\mkern-5mu}\relbar\mapstochar\relbar\relbar}
\def\Leftbararrowfill@{\arrowdoublefill@{\Leftarrow\mkern-2mu}\Relbar\Mapstochar\Relbar\Relbar}
\def\leftringarrowfill@{\arrowdoublefill@{\leftarrow\mkern-3mu}\relbar{\mkern-3mu\circ\mkern-2mu}\relbar\relbar}
\def\lefttriarrowfill@{\arrowfill@{\mathrel\triangleleft\mkern0.5mu\joinrel\relbar}\relbar\relbar}
\def\Lefttriarrowfill@{\arrowfill@{\mathrel\triangleleft\mkern1mu\joinrel\Relbar}\Relbar\Relbar}

\def\hookrightarrowfill@{\arrowfill@{\lhook\joinrel\relbar}\relbar\rightarrow}
\def\twoheadrightarrowfill@{\arrowfill@\relbar\relbar\twoheadrightarrow}
\def\rightbararrowfill@{\arrowdoublefill@{\relbar\mkern-0.5mu}\relbar\mapstochar\relbar\rightarrow}
\def\Rightbararrowfill@{\arrowdoublefill@{\Relbar\mkern-2mu}\Relbar\Mapstochar\Relbar\Rightarrow}
\def\rightringarrowfill@{\arrowdoublefill@\relbar\relbar{\mkern-2mu\circ\mkern-3mu}\relbar{\mkern-3mu\rightarrow}}
\def\righttriarrowfill@{\arrowfill@\relbar\relbar{\relbar\joinrel\mkern0.5mu\mathrel\triangleright}}
\def\Righttriarrowfill@{\arrowfill@\Relbar\Relbar{\Relbar\joinrel\mkern1mu\mathrel\triangleright}}

\def\leftrightarrowfill@{\arrowfill@\leftarrow\relbar\rightarrow}
\def\mapstofill@{\arrowfill@{\mapstochar\relbar}\relbar\rightarrow}

\renewcommand*\xleftarrow[2][]{\ext@arrow 20{20}0\leftarrowfill@{#1}{#2}}
\providecommand*\xLeftarrow[2][]{\ext@arrow 60{22}0{\Leftarrowfill@}{#1}{#2}}
\providecommand*\xhookleftarrow[2][]{\ext@arrow 10{20}0\hookleftarrowfill@{#1}{#2}}
\providecommand*\xtwoheadleftarrow[2][]{\ext@arrow 60{20}0\twoheadleftarrowfill@{#1}{#2}}
\providecommand*\xleftbararrow[2][]{\ext@arrow 10{22}0\leftbararrowfill@{#1}{#2}}
\providecommand*\xLeftbararrow[2][]{\ext@arrow 50{24}0\Leftbararrowfill@{#1}{#2}}
\providecommand*\xleftringarrow[2][]{\ext@arrow 10{26}0\leftringarrowfill@{#1}{#2}}
\providecommand*\xlefttriarrow[2][]{\ext@arrow 80{24}0\lefttriarrowfill@{#1}{#2}}
\providecommand*\xLefttriarrow[2][]{\ext@arrow 80{24}0\Lefttriarrowfill@{#1}{#2}}

\renewcommand*\xrightarrow[2][]{\ext@arrow 01{20}0\rightarrowfill@{#1}{#2}}
\providecommand*\xRightarrow[2][]{\ext@arrow 04{22}0{\Rightarrowfill@}{#1}{#2}}
\providecommand*\xhookrightarrow[2][]{\ext@arrow 00{20}0\hookrightarrowfill@{#1}{#2}}
\providecommand*\xtwoheadrightarrow[2][]{\ext@arrow 03{20}0\twoheadrightarrowfill@{#1}{#2}}
\providecommand*\xrightbararrow[2][]{\ext@arrow 01{22}0\rightbararrowfill@{#1}{#2}}
\providecommand*\xRightbararrow[2][]{\ext@arrow 04{24}0\Rightbararrowfill@{#1}{#2}}
\providecommand*\xrightringarrow[2][]{\ext@arrow 01{26}0\rightringarrowfill@{#1}{#2}}
\providecommand*\xrighttriarrow[2][]{\ext@arrow 07{24}0\righttriarrowfill@{#1}{#2}}
\providecommand*\xRighttriarrow[2][]{\ext@arrow 07{24}0\Righttriarrowfill@{#1}{#2}}

\providecommand*\xmapsto[2][]{\ext@arrow 01{20}0\mapstofill@{#1}{#2}}
\providecommand*\xleftrightarrow[2][]{\ext@arrow 10{22}0\leftrightarrowfill@{#1}{#2}}
\providecommand*\xLeftrightarrow[2][]{\ext@arrow 10{27}0{\Leftrightarrowfill@}{#1}{#2}}

\makeatother

\newcommand{\twocong}[2][0.5]{\ar@{}[#2] \save ?(#1)*{\cong}\restore}
\newcommand{\twoeq}[2][0.5]{\ar@{}[#2] \save ?(#1)*{=}\restore}
\newcommand{\rtwocell}[3][0.5]{\ar@{}[#2] \ar@{=>}?(#1)+/l 0.2cm/;?(#1)+/r 0.2cm/^{#3}}
\newcommand{\ltwocell}[3][0.5]{\ar@{}[#2] \ar@{=>}?(#1)+/r 0.2cm/;?(#1)+/l 0.2cm/^{#3}}
\newcommand{\ltwocello}[3][0.5]{\ar@{}[#2] \ar@{=>}?(#1)+/r 0.2cm/;?(#1)+/l 0.2cm/_{#3}}
\newcommand{\dtwocell}[3][0.5]{\ar@{}[#2] \ar@{=>}?(#1)+/u  0.2cm/;?(#1)+/d 0.2cm/^{#3}}
\newcommand{\dltwocell}[3][0.5]{\ar@{}[#2] \ar@{=>}?(#1)+/ur  0.2cm/;?(#1)+/dl 0.2cm/^{#3}}
\newcommand{\drtwocell}[3][0.5]{\ar@{}[#2] \ar@{=>}?(#1)+/ul  0.2cm/;?(#1)+/dr 0.2cm/^{#3}}
\newcommand{\dthreecell}[3][0.5]{\ar@{}[#2] \ar@3{->}?(#1)+/u  0.2cm/;?(#1)+/d 0.2cm/^{#3}}
\newcommand{\utwocell}[3][0.5]{\ar@{}[#2] \ar@{=>}?(#1)+/d 0.2cm/;?(#1)+/u 0.2cm/_{#3}}
\newcommand{\dtwocelltarg}[3][0.5]{\ar@{}#2 \ar@{=>}?(#1)+/u  0.2cm/;?(#1)+/d 0.2cm/^{#3}}
\newcommand{\utwocelltarg}[3][0.5]{\ar@{}#2 \ar@{=>}?(#1)+/d  0.2cm/;?(#1)+/u 0.2cm/_{#3}}

\newdir{(}{{}*!<0em,-.14em>-\cir<.14em>{l^r}}
\newdir{ (}{{}*!/-5pt/\dir{(}}
\newdir{ >}{{}*!/-5pt/\dir{>}}

\swapnumbers
\theoremstyle{plain}
\newtheorem{Thm}{Theorem}[section]
\newtheorem{Prop}[Thm]{Proposition}
\newtheorem{Cor}[Thm]{Corollary}

\theoremstyle{definition}
\newtheorem{Defn}[Thm]{Definition}

\theoremstyle{remark}
\newtheorem{Ex}[Thm]{Example}

\newcommand{\too}{\mathbf 2}
\begin{document}
 \leftmargini=2em
\title[An abstract view on syntax with sharing]{An abstract view on\\syntax with sharing}
\author{Richard Garner}
\address{Department of Computing, Macquarie University, NSW 2109, Australia}
\email{richard.garner@mq.edu.au} \subjclass[2010]{}
\date{\today}
\begin{abstract}
The notion of term graph encodes a refinement of inductively generated syntax
in which regard is paid to the the sharing and discard of subterms. Inductively
generated syntax has an abstract expression in terms of \emph{initial algebras}
for certain endofunctors on the category of sets, which permits one to go
beyond the set-based case, and speak of inductively generated syntax in other
settings. In this paper we give a similar abstract expression to the notion of
term graph. Aspects of the concrete theory are redeveloped in this setting, and
applications beyond the realm of sets discussed.
\end{abstract}
 \maketitle
\section{Introduction}
A fundamental construction in both mathematics and computer science is the one
which to a signature of finitary operations $\Sigma$ and a set $A$ assigns the
collection $T_\Sigma(A)$ of terms over the signature with free variables in the
set. It provides both the raw syntax out of which semantic structures of all
kinds are constructed, and also the induction principle by which reasoning
about such structures may proceed. In most situations, the passage from the
syntax to the semantics is one of collapse, in which semantically
interchangeable elements of the sets $T_\Sigma(A)$ are identified under a
suitable equivalence relation. However, for some applications we wish to move
in the opposite direction, viewing the sets $T_\Sigma(A)$ themselves as the
extensional collapse of a more intensional structure in which different traces
of execution of the same term can be differentiated from each other. The
concern of this paper is with the particular form of this differentiation in
which regard is paid to the \emph{sharing} and \emph{discard} of computational
values. By this we mean the following: that for a syntactic expression such as
$(\alpha + \beta) * (\alpha + \beta)$, we wish to distinguish the evaluation
path which computes $(\alpha + \beta)$ just once and multiplies the result by
itself, from that which computes it twice, and multiplies the two results
together, from all of those which first go away and performs some irrelevant
computation before beginning the task at hand. Let us for the moment refer to
any means of encoding such distinctions as a \emph{syntax with sharing}. Such
syntaxes have obvious applications to program optimisation; but beyond this,
are important in settings where the execution of functions can incur
\emph{side-effects}---changes in the state of an external environment---which
might render the result of a computation dependent on its execution path.

Here we shall be concerned with a well-established syntax with sharing based
around the notion of \emph{acyclic term
graph}~\cite{Barendregt1987Term,Hasegawa1997Models,Sleep1993Term,Plump1999Term}.
Term graphs generalise the familiar representation of elements of $T_\Sigma(A)$
as well-founded trees, by dropping the requirement that the children of a
particular node be distinct. This is perhaps most easily appreciated through an
example. The term $(\alpha+\beta)*(\alpha+\beta)$ considered above is
represented by the following well-founded tree:
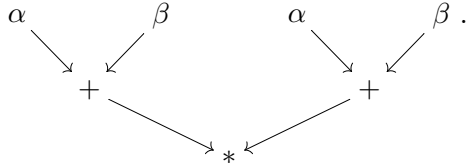
\begin{equation*}
    \cd[@-1em]{
      \alpha \ar[dr] & & \beta \ar[dl] & & \alpha \ar[dr] & & \beta\rlap{ .}  \ar[dl] \\
      & + \ar[drr] & & & & + \ar[dll] \\
      & & & {*} &
}
\end{equation*}
On the other hand, it is represented by any one of a number of different term
graphs, each expressing a computation path with a differing degree of sharing.
One of these term graphs is the same well-founded tree that was just displayed;
and this corresponds to the execution path which computes $(\alpha+\beta)$
twice. On the other hand, the path which computes $(\alpha+\beta)$ only once is
represented by a term graph
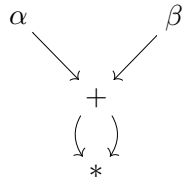
\begin{equation*}
    \cd[@-0.8em]{
      \alpha \ar[dr] & & \beta \ar[dl] \\
      & + \ar@/_8pt/[d] \ar@/^8pt/[d]\\
      & {*}
}
\end{equation*}
whilst the path which does this after first carrying out an irrelevant
computation of $(\alpha + \alpha)$ is represented by
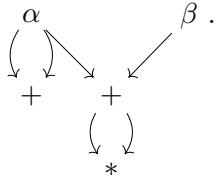
\begin{equation*}
    \cd[@-0.8em]{
      \alpha \ar[dr] \ar@/_8pt/[d] \ar@/^8pt/[d] & & \beta\rlap{ .} \ar[dl] \\
      + & + \ar@/_8pt/[d] \ar@/^8pt/[d]\\
      & {*}
}
\end{equation*}

Now, the construction with which we began this Introduction---that which to a
signature $\Sigma$ and a set $A$ assigns the set of terms $T_\Sigma(A)$---has a
well-known abstract characterisation, achieved by shifting the focus of our
attention away from the signature $\Sigma$, and towards the corresponding
\emph{signature endofunctor}
\begin{equation}\label{eq:signature}\begin{aligned}
    F_\Sigma \colon \cat{Set} & \to \cat{Set}\\
    X & \mapsto \textstyle \sum_{\sigma \in\Sigma} X^{\abs{\sigma}}\ \text,
    \end{aligned}
\end{equation}
wherein we write $\abs{\thg} \colon \Sigma \to \mathbb N$ for the function
assigning arities to each element of the signature. At this level of
generality, the set $T_\Sigma(A)$ may be characterised as an \emph{initial
algebra} for the endofunctor $X \mapsto A + F_\Sigma(X)$ (in a sense recalled
in Definition~\ref{def:initialalg} below). This abstract characterisation of
$T_\Sigma(X)$ also captures its essential structural features---such as the
inductive reasoning it supports---which justifies our interpreting, for an
arbitrary endofunctor $F$ of an arbitrary category $\E$, an initial algebra for
$X \mapsto A + FX$ as being an ``object of terms over $F$ with free variables
in $A$.''

The objective of this paper is to describe a similar abstraction of the notion
of term graph.
We will describe a construction on an endofunctor $F \colon \E \to \E$, which
when applied to a signature endofunctor $F_\Sigma$ on the category of sets,
yields precisely the notion of term graph over $\Sigma$. Moving beyond this
situation allows us to recapture other kinds of term graph; thus taking $\E =
\cat{Set}^S$ for some set $S$, we obtain \emph{many-sorted} term graphs; taking
$\E = \cat{Set}^{\mathbb F}$ (where $\mathbb F$ is the category of finite sets
and bijections) allows us to describe term graphs over second-order syntax;
whilst dualising our construction (in a sense to be made precise later) yields
\emph{cyclic} term graphs, allowing one to capture recursive computations. Of
course, we may also leave the sphere of sets entirely, taking $\E$ to be a
category of domains, or of complete metric spaces, or a topos, or the category
of categories\dots\ The point is that we have a uniform construction providing
us, in each context, with a workable notion of term graph.

We introduce our abstract notion of term graph in Section~\ref{sec:definition}.
From any suitable endofunctor $F$ of a suitable category $\E$, we will
construct a comonad $L_F$ on the arrow category $\E^\too$, whose coalgebras we
define to be abstract term graphs over the endofunctor $F$. We informally
justify our definition by giving a worked example in the category of sets; and
in Section~\ref{sec:equivalence}, make this informal justification precise by
proving that our abstract notion of term graph agrees with the established one
for \emph{any} signature endofunctor of the form~\eqref{eq:signature}.
Section~\ref{sec:generalisations} then considers what the notion of abstract
term graph gives us when we move beyond the motivating set-based case, whilst
Section~\ref{sec:cyclic} shows how \emph{cyclic} term graphs may be captured by
dualising our construction in a particular manner. Finally,
Sections~\ref{sec:interpretation} and~\ref{sec:composition} describe how
further useful aspects of the set-based theory of term graphs may be recaptured
in our abstract setting. In Section~\ref{sec:interpretation}, we show how an
abstract term graph may be interpreted in a suitable semantic domain, whilst in
Section~\ref{sec:composition} we see how abstract term graphs may be composed
into each other.

It is perhaps worth saying a few words about how our abstract treatment of term
graphs is related to others in the literature. One particularly elegant
approach is that described by Hasegawa in his
Ph.~D.~thesis~\cite{Hasegawa1997Models}; with an essentially equivalent one
being given by Corradini and Gadducci
in~\cite{Corradini1997A-2-categorical,Corradini1999An-algebraic}. The key idea there
is to associate with each signature $\Sigma$ a \emph{classifying category}
$\mathscr S[\Sigma]$, whose objects are the natural numbers, and whose
morphisms $n \to m$ are term graphs over $\Sigma$ with $n$ free variables and
$m$ marked output nodes. The structure borne by this category encodes the
various operations on term graphs: with more elaborate kinds of term graph
giving rise to more elaborate kinds of structure on the classifying category.
This approach generalises the notion of term graph in a different direction
from ours; whilst still being tied to the category of sets, it allows one to
impose equations on top of the raw theory of terms. By contrast, our approach
allows one to move beyond the category of sets, but is, as yet, restricted to
freely generated sharing syntaxes. We will see how the two approaches may be
reconciled in Section~\ref{sec:composition}.

A different abstract characterisation of term graphs is described by Hamana
in~\cite{Hamana2010Initial}. In broad strokes the idea is to exploit the linear
representation of term graphs using \textsf{let} syntax; which would, for
example, represent the three term graphs displayed above as
\begin{align*}
    &(\alpha + \beta) * (\alpha + \beta)\text,  \\
    & \mathsf{let}\ z := (\alpha + \beta)\ \mathsf{in}\ z * z \\
    \text{and} \quad &\mathsf{let}\ w := (\alpha + \alpha)\ \mathsf{in}\ \mathsf{let}\ z := (\alpha + \beta)\ \mathsf{in}\ z * z\ \text.
\end{align*}
One may exploit this to give a representation of term graphs based on the
categorical higher-order syntax introduced by Fiore, Plotkin and Turi
in~\cite{Fiore1999Abstract}. This is not precisely what Hamana does, since he
wishes to give an syntax in which inductively defined elements denote term
graphs \emph{uniquely}---something that is not the case for the $\mathsf{let}$
notation; but the general idea should be clear enough. Once again, this
approach differs from ours in being still tied to the category of sets; on the
other hand, it gives a representation of term graphs which is more suitable for
direct implementation.

A third abstract treatment of term graphs, and the one closest in spirit to the
present work, is given in~\cite{Ghani2005Monads}. The idea is to associate to
any endofunctor $F$ a monad $S_F$ for which $S_\Sigma(A)$ is the set of all
term graphs with free variables from $A$ equipped with a marked node
(specifying the output of the computation). However, whilst in principle this
approach allows one to move beyond endofunctors of the category of sets, as the
authors of~\cite{Ghani2005Monads} note in their Section~4.2, there are aspects
of their development that rely on the use of elements, and so cannot be
uniformly generalised beyond the set-based situation.

\textbf{Acknowledgements}. Thanks to Marcelo Fiore, Peter Hancock, Masahito
Hasegawa, Martin Hyland, Alexander Kurz, Paul Blain Levy and Wouter Swierstra
for useful discussions concerning the material of this paper. I acknowledge the
support of St.~John's College, Cambridge, and the Centre for Australian
Category Theory.

\section{Abstract term graphs}\label{sec:definition}
This section describes the construction which underlies our abstract
characterisation of term graphs. Given a suitable endofunctor $F$ on a suitable
category $\E$, it yields a comonad $L_F$ on the arrow category $\E^\too$ whose
coalgebras we shall define to be abstract term graphs over the signature $F$.
We justify this definition in the next section, where we show that this is
literally what they are in the case where $F$ is the endofunctor on the
category of sets associated to a signature $\Sigma$.
First we recall the key abstract notion required for our construction. 

\begin{Defn}\label{def:initialalg}
For any category $\E$ and endofunctor $F \colon \E \to \E$, the category
$F\text-\cat{Alg}$ of \emph{$F$-algebras} has as objects, pairs $(X \in \E, x
\colon FX \to X)$, and as maps $(X, x) \to (Y, y)$, those morphisms $f \colon X
\to Y$ of $\E$ for which $y.Ff = f.x$. An \emph{initial algebra} for $F$ is an
initial object of $F$-$\cat{Alg}$. We denote the underlying object in $\E$ of
such an initial algebra by $\mu X. FX$.
\end{Defn}

Before giving the construction, we motivate it with an example.

\begin{Ex}\label{ex:workedex}
Let $\Sigma$ be the signature $\{\, \alpha, \beta, \mathord +, \mathord *\,\}$
in which $+$ and $*$ are binary operations and $\alpha$ and $\beta$ are
constants, and let $F_\Sigma(X) = X^2 + X^2 + 1+ 1$ be the corresponding
endofunctor on the category of sets. This endofunctor has an initial algebra
$\mu X. F_\Sigma X$ given by the set of closed terms over $\Sigma$. As in the
Introduction, we may represent the elements of this set by well-founded trees,
each of whose nodes is labelled with an element $\sigma \in \Sigma$, and where
each such node has $\abs{\sigma}$ children. We begin by showing how to
generalise this description of closed \emph{terms} to one of closed \emph{term
graphs}, given in terms of the coalgebras for a certain comonad $Q$ on
$\cat{Set}$.

The action on objects of this comonad will be given by $QB = \mu X. B \times
F_\Sigma X$. Thus $QB$ is the set whose elements are well-founded trees over
$\Sigma$ in which every node has also been labelled with an element of $B$. The
action of $Q$ on a function $f \colon B \to B'$ is evident: given a tree $t \in
QB$, we replace the label $b \in B$ at each node by $f(b) \in B'$ to obtain an
element $(Qf)(t) \in QB'$. The counit map $\epsilon_B \colon QB \to B$ is
equally straightforward: for each tree $t \in QB$, we take $\epsilon_B(t)$ to
be the $B$-label of the root. The comultiplication map $\Delta_B \colon QB \to
QQB$ is more subtle. Given a tree over $\Sigma$ labelled in $B$, it must return
a tree over $\Sigma$ labelled in $QB$, and it does so by an operation which we
might describe as ``recursive copying of subtrees''. It is best illustrated
through an example. Let $B = \{a, b, c, d\}$ and consider $t \in QB$ given by
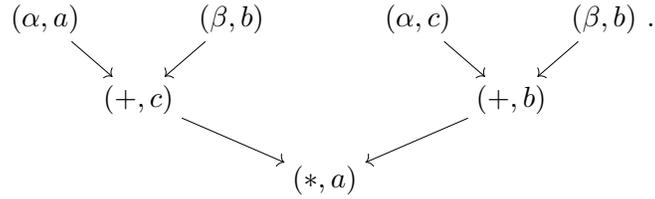
\begin{equation*}
    \cd[@!C@C-2em@R-1em]{
      (\alpha,a) \ar[dr] & & (\beta,b) \ar[dl] & & (\alpha,c) \ar[dr] & & (\beta,b)\rlap{ .}  \ar[dl] \\
      & (\mathord +,c) \ar[drr] & & & & (\mathord +, b) \ar[dll] \\
      & & & {(\mathord *, a)} &
}
\end{equation*}
The tree $\Delta_B(t) \in QQB$ will have the same underlying shape, but the
$B$-label at each node will have been replaced by the $QB$-label given by the
subtree of~$t$ lying above that node. In other words, $\Delta_B(t)$ is the
following tree:
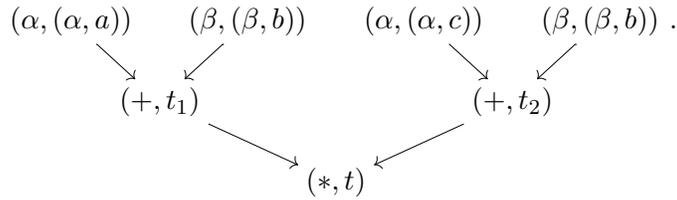
\begin{equation*}
    \cd[@!C@C-4em@R-1em]{
      (\alpha,(\alpha, a)) \ar[dr] & & (\beta,(\beta,b)) \ar[dl] & & (\alpha,(\alpha,c)) \ar[dr] & & (\beta,(\beta,b))\rlap{ .}  \ar[dl] \\
      & (\mathord +,t_1) \ar[drr] & & & & (\mathord +,t_2) \ar[dll] \\
      & & & {(\mathord *, t)} &
}
\end{equation*}
in which $t_1$ and $t_2$ are the respective elements
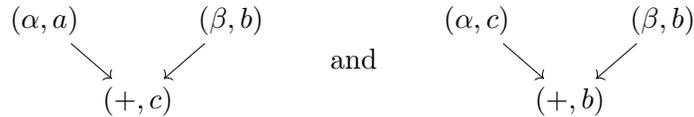
\begin{equation*}
    \cd[@!C@C-2em@R-1em]{
      (\alpha,a) \ar[dr] & & (\beta,b) \ar[dl]\\
      & (\mathord +,c)
} \qquad \text{and} \qquad
\cd[@!C@C-2em@R-1em]{ (\alpha,c) \ar[dr] & & (\beta,b) \ar[dl] \\
 & (\mathord +, b)
}
\end{equation*}
of $QB$. Now a $Q$-coalgebra is given by a set $B$ together with a map $s
\colon B \to QB$ satisfying the two coalgebra axioms. The first of these says
that $\epsilon_B(s(b)) = b$ for all $b \in B$: which is the requirement that
the root of the tree $s(b)$ should be labelled by $b$. The second coalgebra
axiom asks that $\Delta_B(s(b)) = (Qs)(s(b))$ for all $b \in B$: which says
that if a node  of $s(b)$ is labelled by $c \in B$, then the subtree of $s(b)$
lying above that node must coincide with $s(c)$. Our claim is that such
coalgebras correspond with closed term graphs over $\Sigma$. To illustrate
this, consider first a typical closed term graph:
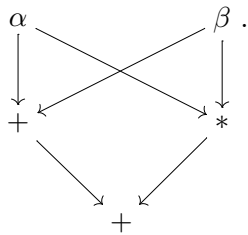
\begin{equation*}
    \cd{
      \alpha \ar[d] \ar[drr] & & \beta\rlap{ .} \ar[dll] \ar[d] \\
      + \ar[dr] & & {*} \ar[dl] \\
      & +}
\end{equation*}
To obtain the corresponding $Q$-coalgebra, we choose a set of labels for the
nodes of this term graph---say $B = \{1, 2, 3, 4, 5\}$ from top to bottom and
left to right---and define a map $s \colon B \to QB$ by taking $s(1) = (\alpha,
1)$, $s(2) = (\beta, 2)$,
\begin{equation*}
    s(3) \ \ = \ \cd[@C-2em@R-1em]{(\alpha,1) \ar[dr] & & (\beta, 2) \ar[dl] \\ & (\mathord +,
    3)}\quad \text, \qquad
    s(4) \ \ = \ \cd[@C-2em@R-1em]{(\alpha,1) \ar[dr] & & (\beta, 2) \ar[dl] \\ & (\mathord *,
    4)}\ \text,
\end{equation*}
and
\begin{equation*}
    s(5) \ \ = \ \cd[@C-2em@R-1em]{(\alpha,1) \ar[dr] & & (\beta, 2) \ar[dl] & & (\alpha,1) \ar[dr] & & (\beta, 2) \ar[dl] \\
    & (\mathord +, 3) \ar[drr] & & & & (\mathord *, 4) \ar[dll]\\
    & & & (\mathord +, 5)}\ \text.
\end{equation*}
This map evidently satisfies the coalgebra axioms. In general, given a closed
term graph over $\Sigma$, we construct a $Q$-coalgebra structure on its set of
nodes as follows. For each node $b$ labelled with $\alpha$ or $\beta$ we set
$s(b) = (\alpha, b)$ or $(\beta, b)$ as appropriate. Then for each node $b$
labelled with $+$ or $*$, and with children $b_1$ and $b_2$, we set $s(b)$ to
be the tree in which the root is $(\mathord +, b)$ or $(\mathord *, b)$ as
appropriate, and the two subtrees of the root are $s(b_1)$ and $s(b_2)$.
Conversely, given any $Q$-coalgebra $s \colon B \to QB$ we define a closed term
graph as follows. Its nodes are the elements of $B$, with each such node $b$
being labelled by that element of $\Sigma$ which labels the root of $s(b)$, and
having as children those elements of $B$ which label the children of the root
of $s(b)$.

In order to capture possibly open term graph over~$\Sigma$, we now describe a
more general comonad $L$; it resides not on $\cat{Set}$ but rather on the arrow
category $\cat{Set}^\too$, the idea being that an $L$-coalgebra structure on an
object $(f \colon A \to B)$ of $\cat{Set}^\too$ should correspond to a term
graph with nodes labelled in $B$ and free variables from the set $A$. In fact,
for each set $A$, the comonad $L$ will restrict and corestrict to the coslice
category $A / \cat{Set}$, yielding a comonad whose coalgebras correspond to
term graphs with free variables from $A$; in particular, taking $A = 0$ we will
recover the earlier comonad $Q$. The underlying functor of the comonad $L$ has
its action on objects given by
\begin{equation*}
    L(f \colon A \to B) = (\ll_f \colon A \to \mu X.A + B \times F_\Sigma X)\ \text,
\end{equation*}
in which $\ll_f$ is defined as follows. Observe that the initial algebra
$\mu X. A + B \times F_\Sigma X$ may be represented as the set of those
well-founded trees built from either nodes labelled in $\Sigma \times B$ as
before, or else leaves labelled only by an element of $A$: under which
representation, the map $\ll_f \colon A \to \mu X. A + B \times F_\Sigma X$
sends $a \in A$ to the tree consisting of the bare leaf $a$. The counit $L
\Rightarrow 1$ of the comonad has as its $f$-component the morphism $\ll_f
\to f$ in $\cat{Set}^\too$ given by
\begin{equation*}
        \cd{
        A \ar[r]^{1_A}\ar[d]_{\ll_f} & A \ar[d]^f \\
        \mu X. A + B \times F_\Sigma X \ar[r]_-{\rr_f} & B
}\end{equation*} wherein $\rr_f$ sends a bare leaf $a$ to $f(a)$ and sends any
other tree to the $B$-label of its root. The comultiplication of $L$ is
analogous to that of $Q$, and we shall not spell it out here. Now, to give an
$L$-coalgebra is to give an object of $\cat{Set}^\too$---which is a function $f
\colon A \to B$---together with a map
\begin{equation*}
        \cd{
        A \ar[r]^{r}\ar[d]_{f} & A \ar[d]^{\ll_f} \\
        B \ar[r]_-{s} & \mu X. A + \rlap{$B \times F_\Sigma X$}
}\end{equation*} in $\cat{Set}^\too$ satisfying the two coalgebra axioms.
It's easy to see that the counit axiom forces $r = 1_A$, whereupon
commutativity of the preceding diagram says that for each $a \in A$, the tree
$s(f(a))$ should be the bare leaf $a$. Note that this in turn forces $f$ to be
a monomorphism. Now the counit axiom says that $\rr_f(s(b)) = b$ for all $b
\in B$. This is trivial for those $b$ in the image of $f$, whilst for those
those that are not, it says that $s(b)$ cannot be a bare leaf $a \in A$ (or
else $b = \rr_f(s(b)) = \rr_f(a) = f(a)$ contradicting $b \notin \mathrm{im}\
f$), and so must as before be a tree whose root is labelled by $(\sigma, b)$
for some $\sigma \in \Sigma$. Finally, the comultiplication axiom says exactly
what it did for $Q$: that if a node of $s(b)$ is labelled by $c \in B$, then
the subtree of $s(b)$ lying above that node must coincide with $s(c)$. What we
now claim is that an $L$-coalgebra such as we have just described corresponds
to a term graph over $\Sigma$ with free variables from $A$. We illustrate this
only with a very simple example. Let $A = \{x, y\}$ and consider the term graph
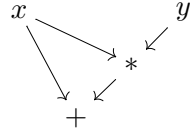
\begin{equation*}
    \cd[@-1.5em]{
      x \ar[drr] \ar[ddr] & & & y \ar[dl] \\
      & & {*} \ar[dl] \\
      & +}
\end{equation*}
with free variables from $A$. One important point to observe is that in our
framework all free variables are maximally shared: which is to say there must
be \emph{exactly one} node in the term graph corresponding to each free
variable. To obtain the $L$-coalgebra corresponding to this term graph, we let
$B = \{x, y, 1, 2\}$ be a labelling of its nodes (including those corresponding
to free variables) and let $f \colon A \hookrightarrow B$ be the evident
inclusion. To define an $L$-coalgebra structure on $f$ we must give a map $s
\colon B \to \mu X. A + B \times F_\Sigma X$ satisfying the appropriate axioms,
which we do by setting $s(x) = x$, $s(y) = y$,
\begin{equation*}
    s(1) \ \ = \ \cd[@C-2.5em@R-1.5em]{x \ar[dr] & & y \ar[dl] \\ & (\mathord +,
    1)}\qquad \text{and} \qquad
    s(2) \ \ = \ \cd[@!C@C-2.5em@R-1.5em]{& x \ar[dr] & & y \ar[dl] \\
    x \ar[dr] & & (\mathord +, 1) \ar[dl]\\
    & (\mathord *, 2)}\ \text.
\end{equation*}
\end{Ex}

This completes our worked example; and we now provide the details of our
construction in its general form. 

\begin{Defn}\label{defn:termgraphcomonad}
Let there be given a category $\E$ with finite products and coproducts, and 
an endofunctor $F \colon \E \to \E$ such that for all 
$A, B \in \E$ the endofunctor $A + B \times F(\thg)$ has an initial algebra.
We define the \emph{term graph comonad} $L_F$ associated
to $F$ as follows.
%
%
%
%
Given an object $f \colon A \to B$ of $\E^\too$, we write $Pf$ for the
initial algebra of the endofunctor $A + B \times F(\thg)$, write
\begin{equation*}
    \iota_f = [\ll_f, \theta_f] \colon A + B \times FPf \to Pf
\end{equation*}
for its algebra structure, and set $L_F(f \colon A \to B) \defeq (\ll_f
\colon A \to Pf)$. To give the action of $L_F$ on a morphism $(h, k) \colon f
\to g$ of $\E^\too$, we set
\begin{equation*}
    L_F\left(\cd{ A \ar[r]^h \ar[d]_f & C \ar[d]^g \\ B \ar[r]_k & D}\right)
\qquad = \qquad  \cd{ A \ar[d]_{\ll_f} \ar[r]^h & C \ar[d]^{\ll_g} \\ Pf \ar[r]_-{P(h, k)} & Pg}
\end{equation*}
where $P(h, k)$ is defined by universality of $Pf$ as the unique map making
\begin{equation*}
    \cd[@C+4em]{
    A + B \times FPf \ar[r]^{A + B \times FP(h,k)} \ar[d]_{\iota_f} &
     A + B \times FPg \ar[d]^{\iota_g . (h + k \times FPg)} \\
    Pf \ar[r]_{P(h,k)} & Pg
    }
\end{equation*}
commute. The
counit and comultiplication natural transformations $L_F \Rightarrow 1$ and $
L_F \Rightarrow L_FL_F$ have as their
respective components at $f \in \E^\too$ the maps $\ll_f \to
f$ and $\ll_f \to \ll_{\ll_f}$ of $\E^\too$ given by:
\begin{equation*}
        \cd{
        A \ar[d]_{\ll_f}\ar[r]^{1_A} & A \ar[d]^{f} \\
        Pf \ar[r]_-{\rr_f} & B
}\qquad \text{and} \qquad
    \cd{
        A \ar[r]^-{1_A}\ar[d]_{\ll_f} & A \ar[d]^{\ll_{\ll_f}} \\
        Pf \ar[r]_-{\sigma_f} & P\ll_f\rlap{ ;}
}
\end{equation*}here $\rr_f$ is defined by universality of $Pf$ as the
unique map making
\begin{equation*}
    \cd[@C+4em]{
    A + B \times FPf \ar[r]^{A + B \times F\rr_f} \ar[d]_{\iota_f} &
     A + B \times FB \ar[d]^{[f, \pi_1]} \\
    Pf \ar[r]_{\rr_f} & B
    }
\end{equation*}
commute, whilst $\sigma_f$ is defined by the same universal property as the
unique map making
\begin{equation*}
    \cd[@C+4em]{
       A + B \times FPf \ar[d]_{\iota_f} \ar[r]^-{A + B \times F\sigma_f} &
        A + B \times FP\ll_f \ar[d]^{[\ll_{\ll_f}, \kappa_f]} \\
        Pf \ar[r]_-{\sigma_f} &
        P\ll_f\rlap{ .}
    }
\end{equation*}
commute, where $\kappa_f$ is defined as the composite
\begin{equation*}
    B \times FP\ll_f \xrightarrow{B \times (F\rr_{\ll_f}, 1)} B \times FPf \times FP\ll_f \xrightarrow{\theta_f \times 1} Pf \times FP\ll_f \xrightarrow{\theta_{\ll_f}} P\ll_f\rlap{ .}
\end{equation*}
\end{Defn}
\begin{Prop}\label{prop:Lcomonad}
The above data determine a comonad $L_F$ on $\E^\too$.
\end{Prop}
\begin{proof}
Entirely routine using the unicity of maps out of an initial algebra.
\end{proof}

\begin{Defn}\label{def:termgraph}
For a category $\E$ with finite products and coproducts, and an endofunctor $F
\colon \E \to \E$ such that each $A + B \times F(\thg)$ has an initial algebra,
we define the category $\cat{ATG}(F)$ of \emph{abstract term graphs over $F$}
to be the category of $L_F$-coalgebras. Explicitly, an abstract term graph over
$F$ is a pair of maps $(f \colon A \to B, s \colon B \to Pf)$ in $\E$
satisfying
    $\rr_f.s = 1_B$,  $s . f = \ll_f$, and $P(1_A, s) .
    s = \sigma_f . s$; whilst a morphism of abstract term graphs $(f, s) \to (g, s')$ is a commutative
    square
\begin{equation*}
\cd{ A \ar[r]^h \ar[d]_f & C \ar[d]^g \\ B \ar[r]_k & D}\end{equation*}
  such that $P(h,k).s = s' . k$.
\end{Defn}
As in Example~\ref{ex:workedex}, the comonad $L$ of our general construction
restricts and corestricts to a comonad on each coslice category $A / \E$, whose
coalgebras are the abstract term graphs over $F$ whose underlying object in
$\E^\too$ has domain $A$. In particular, on taking $A = 0$ we obtain a comonad
$Q \colon \E \to \E$ with $QB = \mu X. B \times FX$; which, as in our example,
we regard as the comonad for \emph{closed} term graphs over $F$. The existence
of the comonad $Q$ was indicated in~\cite{Ghani2001Algebras}, though its
meaning was not discussed; our comonad $L$ may be be seen as a natural
generalisation of it.

\section{Concrete term graphs}\label{sec:equivalence}
In this Section, we show that, by specialising the abstract notion of term
graph given in Definition~\ref{def:termgraph} to the case of a signature
endofunctor $F_\Sigma$ on the category of sets, we recover the usual notion of
acyclic term graph over $\Sigma$. First we give a formal definition of the
latter.
\begin{Defn}\label{def:termgraphconcrete}
A \emph{concrete term graph $T$} over a signature $\Sigma$ is given by:
\begin{itemize}
\item A set of \emph{input nodes} $A$;
\item A set of \emph{internal nodes} $V$;
\item A labelling function $\lbl \colon V \to \Sigma$;
\item For each $v \in V$ and $i \in 1, \dots, \abs{\lbl(v)}$ an element
    $\varphi_i(v) \in A + V$.
\end{itemize}
For such a term graph we define a binary relation on $V$ by $w \lhd v$
    iff $w = \varphi_k(v)$ for some $k$. We say that $T$ is
    \emph{acyclic} if the transitive closure of $\lhd$ is irreflexive, and
    \emph{cyclic} otherwise.
    \end{Defn}
Until further notice we will always interpret the unadorned phrase ``term
graph'' as ``acyclic term graph''.
    \begin{Defn}\label{def:termgraphcat}
If $T$ and $T'$ are concrete term graphs over $\Sigma$, then a \emph{morphism
of term graphs} $T \to T'$ comprises functions $f \colon A \to A'$ and $g
\colon V \to V'$ such that for all $v \in V$ and for all $i$, we have
$\lbl'(g(v)) = \lbl(v)$ and $ (f+g)(\varphi_i(v)) = \varphi'_i(g(v))$.
We write $\cat{CTG}(\Sigma)$ for the category of concrete term graphs over
$\Sigma$.
\end{Defn}
The term graphs of Definition~\ref{def:termgraphconcrete} do not, rightly said,
represent terms so much as computations, since we do not indicate which nodes
should be considered as return values. We may rectify this by adding a set of
\emph{output nodes} $B$ and a labelling function $B \to A + V$ to the
definition: and in Section~\ref{sec:composition} below, we shall. This will
then allow us to compose term graphs by plugging the output nodes of one into
the input nodes of another. However it is the more basic notion that is
pertinent here, as it is the one needed to prove:
\begin{Prop}\label{prop:equivalenceatg}
For any signature $\Sigma$, the categories of concrete term graphs over
$\Sigma$ and of abstract term graphs over $F_\Sigma$ are equivalent.
\end{Prop}
\begin{proof}
Recall that an abstract term graph over $F_\Sigma$ is a coalgebra for the
comonad $L \defeq L_{F_\Sigma}$ of $\cat{Set}^\too$ obtained by the
construction of Section~\ref{sec:definition}. We begin by making explicit the
structure of this comonad. On objects, $L$ sends $f \colon A \to B$ to
$\ll_f \colon A \to Pf$, where $Pf$ is the set defined by the following
inductive clauses:
\begin{itemize}
\item $[a] \in Pf$ for all $a \in A$;
\item $\alpha_b(z_1, \dots, z_{\abs{\alpha}}) \in Pf$ for all $b \in B$,
    $\alpha \in \Sigma$ and $z_1, \dots, z_{\abs \alpha} \in Pf$,
\end{itemize}
and where $\ll_f(a) = [a]$. We introduce the notational convenience of
abbreviating $z_1, \dots, z_{\abs{\alpha}}$ as $\vec z$, and---for any suitable
function $\Gamma$---abbreviating $\Gamma(z_1), \dots, \Gamma(z_{\abs{\alpha}})$
as $\Gamma(\vec z)$. With this notation, the action of $L$ on morphisms
\begin{equation*}
    \cd{A \ar[d]_f \ar[r]^h & C \ar[d]^g \\ B \ar[r]_k & D} \qquad \mapsto \qquad
    \cd{A \ar[d]_{\ll_f} \ar[r]^h & C \ar[d]^{\ll_g} \\ Pf \ar[r]_{P(h,k)} & Pg}
\end{equation*}
may be recursively defined by
\begin{equation*}
    P(h,k)([a]) = [h(a)] \qquad \text{and} \qquad P(h,k)(\alpha_b(\vec z)) = \alpha_{k(b)}(P(h,k)(\vec z))\ \text.
\end{equation*}
Now the map $\rr_f \colon Pf \to B$ giving the counit of $L$ at $f$ is defined
by $\rr_f([a]) = f(a)$ and $\rr_f(\alpha_b(\vec z)) = b$, whilst the map
$\sigma_f \colon Pf \to P\ll_f$ giving the comultiplication at $f$ is
defined recursively by
\begin{equation*}
    \sigma_f( [a] ) = [a] \qquad \text{and} \qquad
    \sigma_f( \alpha_b(\vec z) ) =
    \alpha_{\alpha_b(\vec z)}(\sigma_f(\vec z))\ \text.
\end{equation*}

We will prove the result by defining a functor $F \colon \cat{CTG}(\Sigma) \to
\cat{ATG}(F_\Sigma)$ and showing it to be an equivalence. On objects, the
functor $F$ assigns to each concrete term graph $T = (A, V, \lbl, \varphi)$
the following $L$-coalgebra. Its underlying object in $\cat{Set}^\too$ is
$\mathsf{inl} \colon A \to A + V$. According to Definition~\ref{def:termgraph},
its coalgebra structure is given by a map $s \colon A + V \to P(\mathsf{inl})$,
which will be obtained as follows.
For $a \in A$, we take $s(a) = [a]$. To define $s$ on $V$, we first observe
that since the term graph $T$ is acyclic, the transitive closure $<$ of $\lhd$
is irreflexive, and hence a (strict) partial order on $V$. Moreover, for each
$v \in V$, the set $\{ w\, \mid\, w \lhd v\,\}$ is finite, and hence $\{w
\,\mid\, w < v\}$ is too; we denote its cardinality by $c(v)$. Now given $v \in
V$, suppose we have recursively defined $s(w)$ for all $w \in V$ with $c(w) <
c(v)$. By irreflexivity of $\lhd$, this means in particular that we have
defined $s(w)$ for all $w \lhd v$; and so may validly define
\begin{equation}\label{eq:definingvarphi}
    s(v) = \lbl(v)_v\big(\,s(\varphi_1(v)), \dots, s(\varphi_n(v))\,\big) \qquad \text{(where $n = \abs{\lbl(v)}$).}
\end{equation}
By recursion, this defines $s$ at every $v \in V$. It remains to verify the
coalgebra axioms. It's easy to show that $s.\mathsf{inl} =
\ll_\mathsf{inl}$ and that $\rr_\mathsf{inl} .s = 1_{A+V}$; so it remains
to show that $P(1_A, s) . s = \sigma_f . s$. This is trivial on elements of $A
\subseteq A + V$; whilst to show it on elements of $V \subseteq A + V$ we
proceed by induction. Suppose that $v \in V$ and that we have verified the
equality for all $w$ with $c(w) < c(v)$. Writing $\vec \varphi$ as an
abbreviation for $s(\varphi_1(v)), \dots, s(\varphi_n(v))$, we have
\begin{align*}
    \sigma_f(s(v)) &=
    \sigma_f(\lbl(v)_v(s(\vec \varphi)))\\
    &= \lbl(v)_{s(v)}(\sigma_f(s(\vec \varphi)))\\
    &= \lbl(v)_{s(v)}(P(1, s)(s(\vec \varphi)))\\
    &= P(1,s)(\lbl(v)_v(s(\vec \varphi)))\\
    &= P(1,s)(s(v))
\end{align*}
by the recursive definitions of $\sigma_f$, $P(1, s)$ and~$s$ and the inductive
hypothesis. Hence by induction we have $P(1_A, s).s = \sigma_f . s$ as
required. This completes the definition of the functor $F \colon
\cat{CTG}(\Sigma) \to \cat{ATG}(F_\Sigma)$ on objects. To define it on
morphisms, let $T' = (A', V', \lbl', \varphi')$ be another concrete term
graph, and $(f, g) \colon T \to T'$ a morphism between them. We shall take
$F(f,g) \colon FT \to FT'$ to be given by
\begin{equation}\label{eq:imagefg}
    \cd[@C+1em]{
    A \ar[r]^{f} \ar[d]_{\mathsf{inl}} &
    A' \ar[d]^{\mathsf{inl}} \\
    A+V \ar[r]_-{f+g} &
    A'+V'\rlap{ .}}
\end{equation}
For this to be well defined, we must show that~\eqref{eq:imagefg} is a map of
$L$-coalgebras $FT \to FT'$, i.e., that $P(f,f+g).s = s'.(f+g)$ holds. This is
straightforward on elements of $A \subseteq A + V$; whilst for elements of $V
\subseteq A + V$, we proceed once more by induction. Suppose that $v \in V$ and
that we have verified the equality for all $w \in V$ with $c(w) < c(v)$. Then
we have that
\begin{align}
\nonumber    P(f,f+g) (s(v)) &=
    P(f,f+g)(\lbl(v)_v(s(\varphi_1(v))), \dots, s(\varphi_n(v))))\\
\nonumber    &= \lbl(v)_{gv}(P(f, f+g)(s(\varphi_1(v))), \dots, P(f, f+g)(s(\varphi_1(v))))\\
\label{eq:functorFmor1}    &= \lbl(v)_{gv}(s'(f+g)(\varphi_1(v)), \dots, s'(f+g)(\varphi_n(v)))\ \text,
\end{align}
whilst
\begin{equation}\label{eq:functorFmor2} s'(f+g)(v) = s'(gv) =
\lbl'(gv)_{gv}(s'(\varphi'_1(gv)), \dots, s'(\varphi'_n(gv)))\
\text.
\end{equation}
But since $(f,g)$ is a map of term graphs,~\eqref{eq:functorFmor1}
and~\eqref{eq:functorFmor2} are equal; and so by induction we conclude that
$P(f,f+g).s = s'.(f+g)$ as required. This completes the definition of the
functor $F$; we next show that it is fully faithful. For this, let $T$ and $T'$
be concrete term graphs as before, and suppose that
\begin{equation*}
    \cd[@C+1em]{
    A \ar[r]^{f} \ar[d]_{\mathsf{inl}} &
    A' \ar[d]^{\mathsf{inl}} \\
    A+V \ar[r]_-{h} &
    A'+V'}
\end{equation*}
is a map of $L$-coalgebras $FT \to FT'$. We claim first that $h = f + g$ for
some $g \colon V \to V'$; for which it is clearly enough to show that $h(V)
\subseteq V'$. But were this not so, we would have $h(v) = a'$ for some $v \in
V$ and $a' \in A'$, whence $[a'] = s'(h(v)) = P(f,h)(s(v))$, which is
impossible by the definitions of $P(f,h)$ and $s$. Consequently, every map of
$L$-coalgebras $FT \to FT'$ is of the form~\eqref{eq:imagefg}; and so we will
be done if we can show that for every such map, the pair $(f,g)$ is a map of
concrete term graphs $T \to T'$. But since for every $v \in V$ we have
$P(f,f+g)(s(v)) = s'(f+g)(v)$, equating~\eqref{eq:functorFmor1}
and~\eqref{eq:functorFmor2} shows that that $\lbl(v) = \lbl'(g(v))$ and
that $s'((f+g)(\varphi_i(v))) = s'(\varphi'_i(g(v)))$ for each $1 \leqslant i
\leqslant n$; which since $s'$ is injective, implies that $(f+g)(\varphi_i(v))
= \varphi'_i(g(v))$ for each $i$, so that $(f,g)$ is a map of term graphs as
desired.

Thus $F$ is a fully faithful functor; to complete the proof, we must show that
it is also essentially surjective. So for each $L$-coalgebra $\ell$ we must
find a concrete term graph $T$ and an isomorphism $FT \cong \ell$. From
Definition~\ref{def:termgraph}, to give $\ell$ is to give maps $f \colon A \to
B$  and $s \colon B \to Pf$ in $\E$ satisfying three axioms. The first is that
$s . f = \ll_f$, which says that $s(f(a)) = [a]$ for each $a \in A$. Note
that this forces $f$ to be injective, so that taking $V = B \setminus
\mathrm{im} f$, we have a bijection $B \cong A + V$ under which $f$ is
identified with the coproduct injection $A \hookrightarrow A + V$. The next
coalgebra axiom is that $\rr_f.s = 1_B$, which by case analysis says that
\begin{equation*}
    s(b) = [a] \ \Rightarrow \ b = f(a) \qquad \text{and} \qquad
    s(b) = \alpha_{b'}(\vec z) \ \Rightarrow \ b = b'\ \text;
\end{equation*}
whence $b = f(a)$ if and only if $s(b) = [a]$, so that for $b \in B \setminus
\mathrm{im} f = V$, we must have $s(b) = \alpha_b(\vec z)$ for some $\alpha \in
\Sigma$ and $z_1, \dots, z_{\abs \alpha} \in Pf$. We claim that these $z_i$ in
fact satisfy $z_i = s(\rr_f(z_i))$. Indeed, either $z_i = [a]$ for some $a \in
A$, in which case $s(\rr_f(z_i)) = s(f(a)) = [a] = z_i$; or $z_i =
\beta_c(\vec w)$ for some $\beta$, $c$ and $\vec w$: in which case by the third
coalgebra axiom $P(1,s).s = \sigma_f.s$ we have
\begin{equation*}
    \alpha_{s(b)}(P(1,s)(\vec z)) = P(1,s)(s(b)) =
    \sigma_f(s(b)) = \alpha_{s(b)}(\sigma_f(\vec z))
\end{equation*}
whence $P(1, s)(z_i) = \sigma_f(z_i)$, which implies that
\begin{equation*}
    \beta_{s(c)}(P(1,s)(\vec w)) = P(1, s)(z_i) =
    \sigma_f(z_i) = \beta_{z_i}(\sigma_f(\vec w))
\end{equation*}
so that in particular $z_i = s(c) = s(\rr_f(z_i))$ as claimed. Consequently,
we uniquely determine a function $\lbl \colon V \to \Sigma$ and an
assignation to each $v \in V$ of elements $\psi_1(v), \dots,
\psi_{\abs{\lbl(v)}}(v) \in B$ by the requirement that for all $v \in V$,
\begin{equation}
\label{eq:definingpsi}    s(v) = \lbl(v)_v(s(\psi_1(v)),\dots,s(\psi_n(v))) \qquad \text{(where $n=\abs{\lbl(v)}$),}
\end{equation}
since if $s(v) = \alpha_v(\vec z)$, we necessarily have $\lbl(v) = \alpha$
and $\psi_i(v) = \rr_f(z_i)$ (the second being forced by injectivity of $s$).

We now have a term graph $T = (A, V, \lbl, \varphi)$, where $\varphi$ is
obtained by composing the $\psi$ above with the canonical isomorphism $B \cong
A + V$. It is clear by comparing~\eqref{eq:definingvarphi}
and~\eqref{eq:definingpsi} that $FT$ is isomorphic to the $L$-coalgebra we
started with; and so we will be done as soon as we have checked that $T$ is
acyclic. To do this we consider the function $d \colon Pf \to \mathbb N$
defined by $d([a]) = 0$ and $d(\alpha_b(\vec z)) = \mathrm{max}(d(z_1), \dots,
d(z_n))+1$. Recall that for $v, w \in V$, the relation $v \lhd w$ holds just
when $v = \varphi_k(w)$ for some $k$; but from above, this happens just when
$s(w) = \alpha_b(\vec z)$ and $v = \rr_f(z_k)$ for some $k$, which is equally
well when $s(v) = z_k$ for some $k$. But this implies that $d(v) < d(w)$ in
$\mathbb N$, so that the transitive closure of $\lhd$ must be acyclic as
required.
\end{proof}

\section{Generalisations}\label{sec:generalisations}
Having established the correspondence between concrete term graphs over a
signature and abstract term graphs over the corresponding endofunctor, let us
now see how our abstract notion extends beyond that case. Note that for the
moment, all of our generalisations will remain in the acyclic world; we shall
consider cyclicity in some detail in the following section.

\textbf{Operations with unordered inputs}. To a finitary signature $\abs{\thg}
\colon \Sigma \to \mathbb N$ we can associate an endofunctor $F'_\Sigma$ of
$\cat{Set}$, different from that of~\eqref{eq:signature}, by the formula:
\begin{align}\label{eq:signature2}
    F'_\Sigma(X) = \textstyle \sum_{\sigma \in\Sigma} (X^{\abs{\sigma}} / S_{\abs \sigma})\ \text;
\end{align}
here the set $X^{\abs{\sigma}}$ is being quotiented by that action of the
symmetric group on $\abs{\sigma}$ letters which permutes the order of the
factors. The abstract term graphs generated by such endofunctors are like
concrete term graphs in which the \emph{ordering} of the input variables to an
operation is considered irrelevant.

\textbf{Infinitary operations}. Staying in the category of sets, we can clearly
lift the restriction that the operations in our signature be finitary: indeed,
a signature with infinitary operations still gives rise to a signature
endofunctor by the formula~\eqref{eq:signature} (or~\eqref{eq:signature2} for
that matter).

\textbf{Typed operations}. A \emph{many-sorted signature} is given by a set $S$
of sorts, a set $\Sigma$ of operations, and typing functions for input $i
\colon \Sigma \to \mathbb N^S$ and output $o \colon \Sigma \to S$. Any such
signature generates an endofunctor $F_\Sigma$ of $\cat{Set}^S$ by the formula
\begin{equation*}
   F_\Sigma(\, X_s \mid s \in S\,) = \bigg(\, \sum_{\substack{\sigma \in \Sigma\\o(\sigma) = t}} \prod_{s \in S} (X_s) ^ {i(\sigma, s)}\,\bigg|\, t \in S\, \bigg)\ \text.
\end{equation*}
The corresponding notion of abstract term graph corresponds exactly to the
notion of many-sorted concrete term graph (as defined
in~\cite{Hasegawa1997Models}, for example).

\textbf{Higher-order syntax with sharing}. In~\cite{Fiore1999Abstract}, Fiore,
Plotkin and Turi describe a categorical framework for the study of second-order
syntax. The key idea is to replace the category of sets with the presheaf
category $\cat{Set}^{\mathbb F}$, where $\mathbb F$ denotes the category of
finite cardinals. In this category, initial algebras for endofunctors can be
seen as the collection of terms inductively generated by a second-order
signature, with the value of such an initial algebra at some $n \in \mathbb F$
being the set of all terms with $n$ free variables over the signature. 

This framework was later extended by Tanaka and Power~\cite{Tanaka2000Abstract,Tanaka2005Pseudo-Distributive,Tanaka2006A-Unified} to deal with
more general second-order syntaxes, in which, for example, the second-order entities may be constrained to bind their first-order arguments in a linear fashion. In this more general setting
one still works with a presheaf category, and still describes terms over a second-order signature in terms of initial algebras for suitable endofunctors.

In both situations our construction applies, and so we obtain corresponding  notions of second-order
syntax with sharing. A thorough investigation of this will be a paper in itself
but we hope to convey at least some of what is involved through a simple
example. Consider the following sequent calculus for polynomials over
$\mathbb N$:
\begin{mathpar}
    \inferrule*[right=($1 \leqslant i \leqslant n$)]{\ }{x_1, \dots, x_n \vdash x_i}

    \inferrule*{\ }{x_1, \dots, x_n \vdash 0}

    \inferrule{
      x_1, \dots, x_n \vdash p \\
      x_1, \dots, x_n \vdash q
    }{
      x_1, \dots, x_n \vdash (p + q)}

    \inferrule{
      x_1, \dots, x_n \vdash p \\
      x_1, \dots, x_n \vdash q
    }{
      x_1, \dots, x_n \vdash (p \cdot q)}

    \inferrule*[right=($a \in \mathbb N$)]{
      x_1, \dots, x_n, x_{n+1} \vdash p
    }{
      x_1, \dots, x_n \vdash p[a / x_{n+1}]}
\end{mathpar}
We can organise the terms of this sequent calculus into an object $P \in
\cat{Set}^{\mathbb F}$ in which $Pn$ is the set of all derivable judgements
$x_1, \dots, x_n \vdash p$ and the reindexing function $Pf \colon Pn \to Pm$ is
defined by induction in the obvious way. We may characterise the object $P$ as
the initial algebra for the endofunctor $F$ of $\cat{Set}^\mathbb F$ given by:
\begin{equation*}
    (FX)(n) = n + 1 + (Xn \times Xn) + (Xn \times Xn) + (\mathbb N \times X(n+1))
\end{equation*}
with each term in this sum corresponding to one of the deduction rules listed
above. So we think of elements of $P$ as being closed terms over the signature
$F$ (where here closed is meant in the sense of having no second-order
variables). Applying the constructions of Section~\ref{sec:definition} now yields a corresponding
notion of term \emph{graph}. Without wishing to enter into any detailed calculations, let us at least give an example of what such a term graph will look like. Consider the object $y_3 + y_2 + y_2 + y_3 + y_1$ of $\cat{Set}^\mathbb F$, where $y \colon \mathbb F^\op \to \cat{Set}^\mathbb F$ is the Yoneda embedding. There is a closed term graph structure on this which can represented in \textsf{let} notation as:
\begin{align*}
    x \vdash & \ \mathsf{let}\ p(x,y,z) = y\ \mathsf{in} \\
    & \ \mathsf{let}\ q(x,y) = y\ \mathsf{in} \\
    & \ \mathsf{let}\ r(x,y) = p(x,y,x) + q(y,x)\ \mathsf{in} \\
    & \ \mathsf{let}\ s(x,y,z) = r(x,y) \times r(y,z)\ \mathsf{in} \\
    & \ s(x, x, 48)\ \text.
\end{align*}
It should be clear from this that what such term graphs share
are \emph{second-order} terms: namely the polynomials $p, q, r$ and
$s$. Likewise, when we move from closed term graphs to arbitrary ones, what we
are adding are second-order variables. By way of illustration, we could in the
previous example turn $p$ and $q$ into variables, obtaining a term which in
$\mathsf{let}$ notation would be written as
\begin{align*}
    p(\thg, \thg, \thg), q(\thg, \thg), x \vdash &
    \ \mathsf{let}\ r(x,y) = p(x,y,x) + q(y,x)\ \mathsf{in} \\
    & \ \mathsf{let}\ s(x,y,z) = r(x,y) \times r(y,z)\ \mathsf{in} \\
    & \ s(x, x, 48)\ \text.
\end{align*}
This corresponds to an abstract term graph whose underlying object in
$(\cat{Set}^\mathbb F)^\too$ is given by $y_3 + y_2 \to y_3 + y_2 + y_2 + y_3 + y_1$.

\textbf{Proof theory}. Our final generalisation is very much in the same spirit
as the previous one, and so we do not dwell on the details but merely indicate
a potential application. The categorical proof theory of classical logic is
famously thorny and the hope is that the notion of abstract term graph may
provide an elegant way of encoding some of the computational structure of
classical proofs. The thought is as follows. Starting from some set $V$ of
primitive propositions, we may form $\mathbb F(V)$, the free category with
strictly associative finite products on $V$. Its objects are finite lists $A :=
(A_1, \dots, A_n)$ of elements of $V$ and its morphisms $(A_1, \dots, A_n) \to
(B_1, \dots, B_m)$ are functions $n \to m$ such that $B_{f(i)} = A_i$ for $1
\leqslant i \leqslant m$. Now we can express the collection of classical
proof-trees over the basic propositions in $V$ as an initial algebra for an
endofunctor on the category $\cat{Set}^{\mathbb F V \times V}$ (for a one-sided
sequent calculus) or on the category $\cat{Set}^{\mathbb F V \times (\mathbb F
V)^{\mathrm{op}}}$ (for a two-sided one). Passing to the corresponding notion
of term graph we obtain structures which should allow a smooth representation
of the duplication and discard of sub-proofs which is central to classical
cut-elimination.

\section{Duality and cyclicity}\label{sec:cyclic}
In this section, we describe how \emph{cyclic} term graphs over an endofunctor
$F \colon \E \to \E$ may be captured in our framework. They will also arise as
the coalgebras for a comonad on $\E^\too$; but this comonad will no longer be
obtained by our original construction, but rather by its dual in the following
sense. The endofunctor $F$ is equally well an endofunctor $F^\op \colon \E^\op
\to \E^\op$; and if when we regard it in this way, the hypotheses of Definition~\ref{defn:termgraphcomonad}
are still satisfied (which amounts to the existence of certain \emph{final
coalgebras} in $\E$) then we may apply our construction in $\E^\op$ and regard
the result as structure back in $\E$. Prima facie there is a serious problem
with this, since on the first dualisation we obtain a comonad $L_{F^\op}$ on
$(\E^\op)^\too$, which on the second dualisation becomes a \emph{monad} and not
a comonad on $\E^\too$. We could overcome this if we were to know that the
construction of Definition~\ref{defn:termgraphcomonad} produced not just a comonad, but also at the same time a
monad on $\E^\too$; for then the same would be the true when we passed to the
dual. Remarkably, this is the case; and we now describe this monad explicitly.

\begin{Defn}\label{def:termgraphmonad}
Let there be given a category $\E$ with finite products and coproducts, and 
an endofunctor $F \colon \E \to \E$ such that for all 
$A, B \in \E$ the endofunctor $A + B \times F(\thg)$ has an initial algebra.
We define the \emph{term graph monad} $R_F$ associated
to $F$ as follows. The underlying functor is
given on objects by $R_F(f \colon A \to B) = (\rr_f \colon Pf \to B)$, and on
morphisms by:
\begin{equation*}
    R\left(\cd{ A \ar[r]^h \ar[d]_f & C \ar[d]^g \\ B \ar[r]_k & D}\right)
\qquad = \qquad  \cd{ Pf \ar[d]_{\rr_f} \ar[r]^{P(h,k)} & Pg \ar[d]^{\rr_g} \\ C \ar[r]_-{k} & D\rlap{ .}}
\end{equation*}
Here, $\rho$ and $P$ are given as in Definition~\ref{defn:termgraphcomonad}. 
The transformations $1 \Rightarrow R_F$ and $R_F R_F
\Rightarrow R_F$ making $R_F$ into a monad have their components $f \to \rr_f$
and $\rr_{\rr_f} \to \rr_f$ at some $f \in \E^\too$ given by
\begin{equation*}
\cd{
        A \ar[d]_{f}\ar[r]^-{\ll_f} & Pf \ar[d]^{\rr_f} \\
        B \ar[r]_-{1_B} & B
}
\qquad \text{and} \qquad    \cd{
        P\rr_f \ar[d]_{\rr_{\rr_f}} \ar[r]^-{\pi_f} & Pf \ar[d]^{\rr_f} \\
        B \ar[r]_-{1_B} & B
}
\end{equation*}
respectively. Now $\gamma$ is also as in Definition~\ref{defn:termgraphcomonad}, and the only new datum is the morphism $\pi_f \colon P\rr_f \to
\rr_f$, which we define by the universality of $Pf$ as the unique map
rendering commutative the diagram:
\begin{equation*}
    \cd[@C+4em]{
        Pf + B \times FP\rr_f \ar[r]^{Pf + B \times F\pi_f} \ar[d]_{\iota_{\rr_f}}&
        Pf + B \times FPf \ar[d]^{[1, \theta_f]} \\
        P\rr_f \ar[r]_-{\pi_f} &
        Pf\rlap{ .}
    }
\end{equation*}
\end{Defn}
\begin{Prop}\label{prop:Rmonad}
The above data determine a monad $R_F$ on $\E^\too$.
\end{Prop}
\begin{proof}
Again, all of these are entirely routine calculations with the universal
property of an initial algebra.
\end{proof}

Unwinding the definitions show that to give an $R_F$-algebra structure on a map
$f \colon A \to B$ is to give a morphism $p \colon Pf \to A$ satisfying the
three equations $p.\ll_f = 1_A$, $f.p = \rr_f$ and $p.P(p,1_B) = p.\pi_f$.
In fact, since in giving $p$ we are mapping \emph{out} of an initial algebra,
this description simplifies further.
\begin{Prop}\label{prop:rfcharacterise}
To give an $R_F$-algebra structure on some $(f \colon A \to B) \in \E^\too$ is
equally well to give a map $\phi \colon B \times FA \to A$ satisfying $f.\phi =
\pi_1$; and in these terms, a morphism $(h, k) \colon f \to g$ of $\E^\too$ is
a map of $R_F$-algebras $(f, \phi) \to (g, \phi')$ just when the equation
$h.\phi = \phi'.(k \times Fh)$ is validated.
\end{Prop}
\begin{proof}
For an $R_F$-algebra $p \colon Pf \to A$, the corresponding map $\phi \colon B
\times FA \to A$ over $B$ is given by the composite
\begin{equation*}
    B \times FA \xrightarrow{B \times F\ll_f} B \times FPf \xrightarrow{\theta_f} Pf \xrightarrow{p} A\ \text.
\end{equation*}
Conversely, for a map $\phi \colon B \times FA \to A$ over $B$, the
corresponding $R_F$-algebra structure $p \colon Pf \to A$ is obtained as the
unique map making the square
\begin{equation*}
    \cd[@C+4em]{
    A + B \times FPf \ar[r]^{A + B \times Fp} \ar[d]_{\iota_f} &
     A + B \times FA \ar[d]^{[1_A, \phi]} \\
    Pf \ar[r]_{p} & A
    }
\end{equation*}
commute. The remaining verifications are straightforward.
\end{proof}
Just as the comonad $L_F$ induces a comonad on each coslice category $A / \E$,
so $R_F$ induces a monad on each slice category $\E / B$. In particular, when
$B = 1$, we obtain the monad on $\E$ whose underlying assignation on objects is
given by $A \mapsto \mu X. A + FX$. We recognise this as the \emph{free monad}
on the endofunctor $F$, which is characterised by the property that its
category of algebras is canonically isomorphic to the category of $F$-algebras
in the sense of Definition~\ref{def:initialalg}. The monad $R_F$ may be seen as
a generalisation of this, with the preceding Proposition being the
corresponding generalisation of the universal property of the free monad.

We now give the promised dualisation of the constructions of Definitions~\ref{defn:termgraphcomonad} and \ref{def:termgraphmonad}.
\begin{Defn}
Let $\E$ be a category with finite products and coproducts, and $F \colon \E \to \E$
an endofunctor such that for all $A, B \in \E$, the endofunctor $B \times (A + F\thg)$ 
admits a final coalgebra. Then we define the \emph{cyclic term graph comonad} 
to be $\bar L_F \defeq (R_{F^\op})^\op$, and the \emph{cyclic term graph monad} to be 
to be $\bar R_F \defeq (L_{F^\op})^\op$.
\end{Defn}

Before going on, let us extract an explicit description of the comonad $\bar
L_F$. Given an object $f \colon A \to B$ of $\E^\too$, we write $\bar P f$ for
the final coalgebra of $X \mapsto B \times (A +FX)$, write
\begin{equation*}
    {\bar \iota}_f = ({\bar \rr}_f, {\bar s}_f) \colon \bar Pf \to B \times (A + F\bar Pf)
\end{equation*}
for its coalgebra structure, and write $\bar \ll_f \colon A \to \bar Pf$
for the unique map making the square
\begin{equation*}
    \cd[@C+4em]{
    A \ar[d]_{(f, \mathsf{inl})} \ar[r]^{{\bar \ll}_f} &
    \bar Pf \ar[d]^{{\bar \iota}_f} \\
    B \times (A + FA) \ar[r]_{B \times (A + F\bar \ll_f)} &
     B \times (A + F\bar Pf)
    }
\end{equation*}
commute. We now define $\bar L_F$ on objects by $\bar L_F(f \colon A \to B)
\defeq (\bar \ll_f \colon A \to Pf)$. We define its action on morphisms $(h, k) \colon f \to g$
of $\E^\too$ by
\begin{equation*}
    \bar L_F\left(\cd{ A \ar[r]^h \ar[d]_f & C \ar[d]^g \\ B \ar[r]_k & D}\right)
\qquad = \qquad  \cd{ A \ar[d]_{\ll_f} \ar[r]^h & C \ar[d]^{\ll_g} \\ \bar Pf \ar[r]_-{\bar P(h, k)} & \bar Pg}
\end{equation*}
where $\bar P(h, k)$ is defined by universality of $\bar Pg$ as the unique map
making
\begin{equation*}
    \cd[@C+4em]{
    \bar Pf \ar[r]^{\bar P(h,k)} \ar[d]_{(k \times (h + F\bar P f)). \bar \iota_f} &
     \bar Pg \ar[d]^{\bar \iota_g} \\
    D \times (C + F\bar Pf) \ar[r]_{D \times (C \times F\bar P(h,k))} & D \times (C + F\bar Pg)
    }
\end{equation*}
commute. The natural transformations $\bar L_F \Rightarrow
1$ and $\bar L_F \Rightarrow \bar L_F\bar L_F$ providing the
comonad structure have respective $f$-components given by maps
\begin{equation*}
        \cd{
        A \ar[d]_{\bar \ll_f}\ar[r]^{1_A} & A \ar[d]^{f} \\
        \bar Pf \ar[r]_-{\bar \rr_f} & B
}\qquad \text{and} \qquad
    \cd{
        A \ar[r]^-{1_A}\ar[d]_{\bar \ll_f} & A \ar[d]^{\bar \ll_{\bar \ll_f}} \\
        \bar Pf \ar[r]_-{\bar \sigma_f} & \bar P\bar \ll_f
}
\end{equation*} in $\E^\too$, where $\bar \rr_f$ is defined as above, and $\bar \sigma_f$ is defined by universality of
$\bar P \bar \ll_f$ as the unique map rendering commutative the square:
\begin{equation*}
    \cd[@C+4em]{
        \bar Pf \ar[r]^{\bar \sigma_f} \ar[d]_{(1, \bar \theta_f)}&
        \bar P\bar \ll_f \ar[d]^{\bar \iota_{\bar \ll_f}} \\
        \bar Pf \times (A + F\bar Pf) \ar[r]_-{\bar Pf \times (A + F\bar \sigma_f)} &
        \bar Pf \times (A + F\bar P\bar \ll_f)\rlap{ .}
    }
\end{equation*}

\begin{Defn}\label{def:cyctermgraph}
For a category $\E$ with finite products and coproducts, and an endofunctor $F
\colon \E \to \E$ such that each $B \times (A + F(\thg))$ has a final
coalgebra, we define the category $\cat{ATG}_\infty(F)$ of \emph{cyclic
abstract term graphs over $F$} to be the category of $\bar L_F$-coalgebras.
\end{Defn}
By the dual of Proposition~\ref{prop:rfcharacterise}, the category
$\cat{ATG}_\infty(F)$ is isomorphic to the category whose objects are pairs $(f
\colon A \to B,\, s \colon B \to A + FB)$ for which $s.f = \mathsf{inl}$, and
whose morphisms $(f, s) \to (g, s')$ are maps $(h, k) \colon f \to g$ in
$\E^\too$ for which $s'.k = (h + Fk).s$; which is almost precisely the
definition of cyclic term graph given in~\cite{Barendregt1987Term}. Note that
if we also took this as our definition of cyclic abstract term graphs, then it
would make sense under much weaker hypotheses than those of
Definition~\ref{def:cyctermgraph}: it is enough that $\E$ should have binary
coproducts. Although this extra generality is certainly useful, for the present
paper we shall retain the narrower definition, and this for two reasons:
firstly, to highlight the duality between the cyclic and the acyclic cases; and
secondly, so that later on, when we consider further aspects of the theory, we
can treat these two cases in a uniform manner.

Let us now show that abstract cyclic term graphs are a faithful generalisation
of the concrete ones.
\begin{Prop}\label{prop:equivalencecycatg}
For any signature $\Sigma$, the categories of cyclic concrete term graphs over
$\Sigma$ and of cyclic abstract term graphs over $F_\Sigma$ are equivalent.
\end{Prop}
\begin{proof}
The method of proof is the same as Proposition~\ref{prop:equivalenceatg}: we
define a functor $F \colon \cat{CTG}_\infty(\Sigma) \to
\cat{ATG}_\infty(F_\Sigma)$ and show it to be an equivalence. On objects, given
a cyclic concrete term graph $T = (A, V, \lbl, \varphi)$, we observe that
$\lbl$ and $\varphi$ together determine a morphism $\ell \colon V \to
F_\Sigma(A)$; so that we may take $F(T)$ to be the $\bar L_F$-coalgebra whose
underlying object in $\cat{Set}^\too$ is $\mathsf{inl} \colon A \to A + V$, and
whose coalgebra structure corresponds under the isomorphism of
Proposition~\ref{prop:rfcharacterise} to the map $A + \ell \colon A + V \to A +
F_\Sigma(A + V)$. The remaining details are entirely analogous to
Proposition~\ref{prop:equivalenceatg} (though simpler) and hence omitted.

(Observe that when we apply Proposition~\ref{prop:rfcharacterise} here, we are
really doing something quite familiar. Turning the map $A + \ell$ into an $\bar
L_{F_\Sigma}$-coalgebra structure on $\mathsf{inl} \colon A \to A + V$
corresponds to taking the concrete cyclic term graph $T$ and unfolding it into
a possibly-infinite labelled tree.)
\end{proof}

It is quite straightforward to see that for each of the more general examples
discussed in Section~\ref{sec:generalisations}, applying the dual construction
yields an appropriate notion of cyclic term graph.

\section{Interpretation}\label{sec:interpretation}
Now that we have good abstract notions of both acyclic and cyclic term graph,
we wish to develop further aspects of their theory. In this section, we discuss
how to interpret term graphs in a suitable semantic domain; whilst in the next,
we shall discuss how abstract term graphs may be composed. In order to give a
uniform treatment of both kinds of term graph, we shall describe a general
structure of which both are particular instances; so that by framing subsequent
results in terms of this general structure, we may deal with both cases
simultaneously. The structure in question is not an \emph{ad hoc} one, but one
of importance in abstract homotopy theory and category theory.
\begin{Defn}
A \emph{natural weak factorisation
system}~\cite{Grandis2006Natural,Garner2009Understanding} on a category $\E$ is
given by the assignation of a factorisation
\begin{align}
\label{eq:functfact}   \cd{A \ar[r]^f & B} \qquad \ &\mapsto \,\ \qquad \cd{ A \ar[r]^{\ll_f} & Pf \ar[r]^{\rr_f} & B }
\intertext{to every morphism of $\E$; a factorisation}
\label{eq:functsquare}  \cd{ A \ar[r]^f \ar[d]_h & B \ar[d]^k \\ C \ar[r]_g & D}
\qquad &\mapsto \qquad  \cd{ A \ar[r]^{\ll_f} \ar[d]_h & Pf \ar[r]^{\rr_f} \ar[d]^{P(h, k)} & B \ar[d]^k \\ C \ar[r]_{\ll_g} & Pg \ar[r]_{\rr_g} & D}
\end{align}
to every commutative square of $\E$, functorial in $(h, k)$; and for each $f
\colon A \to B$ in $\E$, choices of maps $    \sigma_f \colon Pf \to
P\ll_f$ $\pi_f \colon P\rr_f \to Pf$ such that:
\begin{itemize}
\item There is a comonad $(L, \epsilon, \Delta)$ on $\E^\too$ with $Lf =
    \ll_f$, with $\epsilon_f = (1, \rr_f) \colon \ll_f \to f$ and
    with $\Delta_f = (1, \sigma_f) \colon \ll_f \to
    \ll_{\ll_f}$.
\item There is a monad $(R, \eta, \mu)$ on $\E^\too$ with $Rf = \rr_f$,
    with $\eta_f = (\ll_f, 1) \colon f \to \rr_f$ and with $\mu_f =
    (\pi_f, 1) \colon \rr_{\rr_f} \to \rr_f$.
\item There is a distributive law $\gamma \colon LR \Rightarrow RL$ whose
    component at $f$ is given by $\gamma_f = (\sigma_f, \pi_f) \colon
    \ll_{\rr_f} \to \rr_{\ll_f}$.
\end{itemize}
\end{Defn}
The notion of natural weak factorisation system is a strengthening of Quillen's
notion of \emph{weak factorisation system}~\cite{Quillen1967Homotopical}, which
has found use in computer science in the open map approach to bisimulation
of~\cite{Joyal1996Bisimulation,Cattani2005Profunctors}.

\begin{Prop}\label{prop:havewfs}
For any category $\E$ with finite products and coproducts, and any endofunctor
$F \colon \E \to \E$ for which each initial algebra $\mu X. A + B \times FX$
exists, the monad-comonad pair $(L_F, R_F)$ on $\E^\too$ is a natural weak
factorisation system.
\end{Prop}
\begin{proof}
All that remains is to exhibit the required distributive law $\gamma$. Since we
already have the necessary data, we need only check the corresponding axioms,
which we may do through a further straightforward manipulation using the
universal property of an initial algebra.
\end{proof}
By duality, we immediately obtain:
\begin{Cor}\label{prop:havewfs2} For any
category $\E$ with finite products and coproducts, and any endofunctor $F \colon
\E \to \E$ for which each final coalgebra $\nu X. B \times (A + FX)$ exists,
the monad-comonad pair $(\bar L_F, \bar R_F)$ on $\E^\too$ is a natural weak
factorisation system.
\end{Cor}
The two preceding results are more than a convenient framing device for our
subsequent development: they actually guide that development, by allowing us to
apply aspects of the theory of natural weak factorisation systems to the study
of term graphs. For our first such application, we derive a notion of
\emph{interpretation} for abstract term graphs, using the following basic
result from the theory of natural weak factorisation systems:
\begin{Prop}[``Lifting'']\label{prop:lifting} If $(L, R)$ is a natural weak factorisation system on a
category $\E$, then for any commutative square
\begin{equation*}
    \cd{
      A \ar[r]^h \ar[d]_f & C \ar[d]^g \\
      B \ar[r]_k & D
    }
\end{equation*}
in $\E$, any $L$-coalgebra structure on $f$ and any $R$-algebra structure on
$g$, there is a canonical choice of morphism $j \colon B \to C$ such that $gj =
k$ and $jf = h$.
\end{Prop}
\begin{proof}
The $L$-coalgebra structure on $f$ is given by a morphism $s \colon B \to Pf$
satisfying axioms; likewise, the $R$-algebra structure on $g$ is given by a map
$p \colon Pg\to C$. We may therefore take $j$ to be the composite
\begin{equation}\label{eq:liftcomposite}
    B \xrightarrow{s} Pf \xrightarrow{P(h,k)} Pg \xrightarrow{p} C\ \text. \qedhere
\end{equation}
\end{proof}
Let us see how this pertains to term graphs. We shall specialise
Proposition~\ref{prop:lifting} to the particular case where $D = 1$, and apply
it first to the acyclic situation of Proposition~\ref{prop:havewfs}, and then
to the cyclic situation of Corollary~\ref{prop:havewfs2}. In the former case,
the basic data we have is a diagram
\begin{equation}\label{eq:interpretationbasicdata}
    \cd{A \ar[d]_{f} \ar[r]^h & C \\ B}
\end{equation}
where $f$ is an $L_F$-coalgebra---hence an acyclic term graph over $F$---and
the unique map $C \to 1$ is an $R_F$-algebra; which by the discussion following
Proposition~\ref{prop:rfcharacterise}, is equally well to say that $C$ bears an
$F$-algebra structure $c \colon FC \to C$. Now the object $A$ is the object of
free variables of the acyclic term graph $f$; and so to give the map $h$ is to
give an interpretation of these variables in the $F$-algebra $C$. The canonical
map $j \colon B \to C$ whose existence is assured by
Proposition~\ref{prop:lifting} extends this to an interpretation of \emph{all}
nodes of the given term graph in $C$, and does so using the $F$-algebra
structure in the obvious manner. In fact, by unwinding the definitions
in~\eqref{eq:liftcomposite}, we find that the map $j$ specified there is
obtained by composing the coalgebra map $s \colon B \to Pf$ with the map
$\mathsf{ev} \colon Pf \to C$ obtained by universality in
\begin{equation*}
    \cd[@C+3.5em]{
        A + B \times FPf \ar[r]^{A + B \times F(\mathsf{ev})} \ar[d]_{\iota_f} &
        A + B \times FC \ar[d]^{[h, c\pi_2]} \\
        Pf \ar[r]_{\mathsf{ev}} & C\rlap{ ;}
    }
\end{equation*}
which is easily seen to agree with the natural notion of interpretation we
would give in the concrete situation. Let us now consider the cyclic case. This
time, our basic data are a diagram of the
form~\eqref{eq:interpretationbasicdata}, an $\bar L_F$-coalgebra structure on
$f$---making it into a cyclic term graph---and an $\bar R_F$-algebra structure
on $C \to 1$. Now, to give the latter is equally well to equip $C$ with an
algebra structure for the monad obtained on $\E \cong \E / 1$ by restricting
and corestricting $R_F$ to those objects of $\E^\too$ with codomain $1$. The
monad in question is the one whose underlying assignation on objects is given
by
\begin{equation*}
    A \mapsto \nu X. A + FX\ \text;
\end{equation*}
it has been studied carefully in~\cite{Aczel2003Infinite,Milius2005Completely},
where it is called the \emph{free completely iterative monad} on $F$. Its
algebras are called \emph{completely iterative $F$-algebras}, and are
characterised as being those $F$-algebras with the property that every system
of guarded recursive equations over $F$ has a solution. Without going into the
details of this let us merely say that this is precisely what is captured by
our notion of interpretation. We may regard the cyclic term graph $f \colon A
\to B$ as a system of guarded recursive equations over $F$ with constants in
the set $A$. The map $h \colon A \to C$ indicates how to interpret the
constants of the recursive equations in the completely iterative $F$-algebra
$C$; whilst the extension to a map $j \colon B \to C$ provides the
corresponding solution.

\section{Composition}\label{sec:composition}
In this final section, we shall show that our abstract term graphs admit an
operation of \emph{composition}, which chains the results of computation from
one term graph to another. In order to perform such a chaining, an extra datum
is required, indicating how the free variables of the second term graph should
be filled by values from the computation of the first. As presaged in the
discussion following Definition~\ref{def:termgraphcat}, we shall determine this
extra datum by considering term graphs equipped with a distinguished collection
of ``output nodes''.
\begin{Defn}\label{def:termgraphfrom}
Let $F \colon \E \to \E$ be an endofunctor to which the construction of
Section~\ref{sec:definition} (respectively, its dual) applies. An
\emph{acyclic} (respectively, \emph{cyclic}) \emph{term graph from $A$ to $B$}
over $F$ is a cospan
\begin{equation*}
    \cd{A \ar[dr]_f & & B \ar[dl]^g \\ & X}
\end{equation*}
together with an $L_F$- (respectively, $\bar L_F$-) coalgebra structure  on
$f$.
\end{Defn}
In the concrete case, we see by Proposition~\ref{prop:equivalenceatg} that
acyclic or cyclic term graphs from $A$ to $B$ over a signature endofunctor
$F_\Sigma$ correspond to concrete term graphs $T = (A, V, \lbl, \varphi)$
equipped with a function $g \colon B \to A + V$. In that case, we may compose
such a pair $(T, g) \colon A \to B$ with another pair $(T', g') \colon B \to
C$, to obtain the pair $(T'\circ T, h) \colon A \to C$ given as follows.
\begin{itemize}
\item The set of input nodes of $T' \circ T$ is $A$ (as it must be);
\item The set of internal nodes is $V + V'$;
\item The labelling function is $[\lbl, \lbl'] \colon V + V' \to
    \Sigma$;
\item The children of an element $v \in V$ are given by $\varphi_i(v)$;
\item The children of an element $v' \in V'$ are given by
    \begin{equation*}
        \psi_i(v') = \begin{cases} \varphi'_i(v) & \text{if $\varphi'_i(v) \in V'$}\\
        g(\varphi'_i(v)) & \text{if $\varphi'_i(v) \in B$}
        \end{cases}
    \end{equation*}
\item The function $h \colon C \to A + V + V'$ is given by
\begin{equation*}
  h(c) =
        \begin{cases} g'(c) & \text{if $g'(c) \in V'$}\\
        g(g'(c)) & \text{if $g'(c) \in B$\ \text.}\end{cases}
\end{equation*}
\end{itemize}
In the acyclic case, the required acyclicity of the composite follows from that
of the two parts and a case analysis. What we shall now do is provide an
abstract analogue of this composition.
\begin{Prop}\label{prop:sharingtheory}
Let $F \colon \E \to \E$ be an endofunctor to which the construction of
Section~\ref{sec:definition} (respectively, its dual) applies, and let $\E$
have pushouts. Under these hypotheses there is a category $\mathscr{S}[F]$
(respectively $\mathscr S_\infty[F]$) whose objects are those of $\E$ and whose
morphisms $A \to B$ are equivalence classes of acyclic (respectively cyclic)
term graphs from $A$ to $B$.
\end{Prop}
The notion of equivalence we use in this Proposition identifies two term graphs
$A \to B$ just when there is an isomorphism $\lambda \colon X \to X'$ making
\begin{equation*}
    \cd[@C+2em]{A \ar[dr]^f \ar[ddr]_{f'} & & B \ar[ddl]^{g'} \ar[dl]_g \\ & X \ar[d]^\lambda \\ & X'}
\end{equation*}
commute, and making the left-hand triangle a map of coalgebras for the
appropriate comonad. The reason for quotienting in this way is that we intend
to define the composition of two cospans $A \rightarrow X \leftarrow B$ and $B
\rightarrow Y \leftarrow C$ by taking it to be the outer edge of the diagram
\begin{equation}\label{eq:pushoutcomp}
    \cd[@!@-1em]{
    A \ar[dr]_f & & B \ar[dl]_g \ar[dr]^h & & C\rlap{ ,}\ar[dl]^k \\
    & X \ar[dr]_{p} & & Y \ar[dl]^{q} \\ & & Z
    }
\end{equation}
wherein the bottom square is a pushout. As it stands, this composition is only
associative up to isomorphism: and to rectify this, we must quotient out as
above. This could be avoided if we were to make $\cat{Cospan}(F)$ into a
\emph{bicategory} rather than a category, but for our purposes, passing to the
quotient seems to be the simplest way to proceed.

The other obstacle to defining the composition as in~\eqref{eq:pushoutcomp}
lies in showing that the given coalgebra structures on $f$ and $h$ induce one
on $pf$. We shall do this in two stages: first we show that ``a pushout of an
coalgebra is a coalgebra''---which gives us a coalgebra structure on $p$ from
the one on $h$---and then we show that ``the composite of two coalgebras is a
coalgebra''---which gives us the one on $pf$ from those on $p$ and on $f$. We
may prove these results for the acyclic and the cyclic cases simultaneously, as
they are completely general facts about natural weak factorisation systems.

\begin{Prop}[``$L$-coalgebras push out'']\label{prop:coalgpushout}
Let $(L, R)$ be a natural weak factorisation system  on a category $\E$. For
any pushout square
\begin{equation*}
  \cd{A \ar[d]_f \ar[r]^h & C \ar[d]^g \\ B \ar[r]_k & D}
\end{equation*}
and any $L$-coalgebra structure on $f$, there is a unique $L$-coalgebra
structure on $g$ making the square a map of $L$-coalgebras.
\end{Prop}
\begin{proof}
To give an $L$-coalgebra structure on $f$ is to give a map $s \colon B \to Pf$
satisfying $\rr_f . s = 1_B$, $s . f = \ll_f$ and $P(1_A,s).s = \sigma_f .
s$. We induce a corresponding map $t \colon D \to Pg$ for $g$ by applying the
universal property of pushout to the square
\begin{equation*}
  \cd[@C+2em]{A \ar[d]_f \ar[r]^h & C \ar[d]^{\ll_g} \\ B \ar[r]_{P(h, k).s} & Pg\rlap{ .}}
\end{equation*}
Thus $t$ is the unique map $D \to Pg$ satisfying $t.g = \ll_g$ and $t.k =
P(h,k).s$, and so will be the unique $L$-coalgebra structure on $g$ making
$(h,k)$ into a map of $L$-coalgebras as soon as we have verified the other two
$L$-coalgebra axioms: which is easy by the universal property of pushout.
\end{proof}

\begin{Prop}[``$L$-coalgebras compose'']\label{prop:coalgcompose}
Let $(L, R)$ be a natural weak factorisation system on a category $\E$ and let
$f \colon A \to B$, $g \colon B \to C$ in $\E$. For every choice of
$L$-coalgebra structures on $f$ and $g$, there is a unique compatible
$L$-coalgebra structure on $gf$.
\end{Prop}
By a \emph{compatible} $L$-coalgebra structure on $gf$, we mean the following.
By virtue of the given coalgebra structures on $f$ and $g$, we have for any
square
\begin{equation*}
    \cd{
    A \ar[d]_{gf} \ar[r]^h & D \ar[d]^p \\
    C \ar[r]_{k} & E
    }
\end{equation*}
and any $R$-algebra structure on $p$, a choice of filler $j \colon C \to D$
obtained by applying Proposition~\ref{prop:lifting} twice: first with $f$ on
the left, and then with $g$. An $L$-coalgebra structure on $gf$ is compatible
if the preceding choices of fillers agree with those obtained by applying
Proposition~\ref{prop:lifting} once to $gf$.

\begin{proof}
For uniqueness, we observe that any given $L$-coalgebra structure on $gf$ may
be recovered by applying Proposition~\ref{prop:lifting} to the square
\begin{equation*}
    \cd{
        A \ar[r]^-{\ll_{gf}} \ar[d]_{gf} & P(gf) \ar[d]^{\rr_{gf}} \\
        C \ar[r]_{1_C}  & C
    }
\end{equation*}
where $\rr_{gf}$ is given its free $R$-algebra structure. Thus there can be at
most one compatible $L$-coalgebra structure on $gf$, and we can calculate what
it must be by applying Proposition~\ref{prop:lifting} twice to the above
square, first with $f$ along the left and then with $g$. Let the $L$-coalgebra
structures on $f$ and $g$ be given by $s \colon B \to Pf$ and $t \colon C \to
Pg$ respectively. Then direct calculation shows that the induced $L$-coalgebra
structure on $gf$ is given as follows. First form the composite
\begin{equation*}
    \xi = B \xrightarrow{s} Pf \xrightarrow{P(1,g)} P(gf)\ \text.
\end{equation*}
Now the induced $L$-coalgebra structure $C \to P(gf)$ is given by the composite
\begin{equation*}
    C \xrightarrow{t} Pg \xrightarrow{P(\xi, 1_C)} P\rr_{gf} \xrightarrow{\pi_{gf}} P(gf)\ \text.
\end{equation*}
That this is indeed an $L$-coalgebra structure, and a compatible one, is
straightforward calculation.
\end{proof}
Applying the preceding two results to the natural weak factorisation systems of
Proposition~\ref{prop:havewfs} and Corollary~\ref{prop:havewfs2}, we obtain:
\begin{proof}[Proof of Proposition~\ref{prop:sharingtheory}]
As anticipated, we define composition in $\mathscr S[F]$ and $\mathscr
S_\infty[F]$ by pushouts of the form~\eqref{eq:pushoutcomp}, using the
preceding two Propositions to induce the required coalgebra structure on the
composite left leg from the coalgebra structures on the constituents. We give
the identity map at $A$ by the cospan
\begin{equation*}\cd{
    A \ar[dr]_{1_A} & & A \ar[dl]^{1_A} \\ & A}
\end{equation*}
where the left leg is equipped with its unique possible coalgebra structure. We
must check that this composition is associative and unital. Because we have
passed to equivalence classes, we have this at the level of underlying cospans;
it remains to verify that the induced coalgebra structures on the composites
are likewise well-behaved. But this follows easily from the universal
properties ascribed to the constructions of
Propositions~\ref{prop:coalgpushout} and~\ref{prop:coalgcompose}.
\end{proof}

Using the results of Proposition~\ref{prop:equivalenceatg} we may now show that
the composition of abstract term graphs, when specialised to a signature
endofunctor on $\cat{Set}$, agrees with the composition described after
Definition~\ref{def:termgraphfrom}. Note that this in particular provides an
abstract reason why this latter composition should be associative, a fact which
would otherwise have required a direct calculation.

Let us remark that the categories $\mathscr S[F]$ and $\mathscr S_\infty[F]$
will not typically be locally small, even if $\E$ is so. The reason is
essentially that a term graph from $A$ to $B$ may choose to do an arbitrary
amount of irrelevant computation which is invisible from the perspective of the
output nodes in $B$. Thus in practice it may be convenient to consider a
suitable full small subcategory $\A \subset \E$, and to cut down from $\mathscr
S[F]$ to the subcategory whose objects are those lying in $\A$, and whose
morphisms $A \rightarrow X \leftarrow B$ are those cospans where $A$, $X$ and
$B$ all lie in $\A$ (where for this definition to work we must assume that $\A$
is closed under the appropriate pushouts in $\E$). Thus when $\E = \cat{Set}$
and $F$ is the endofunctor associated to a signature $\Sigma$, we may take $\A$
to be the category of finite cardinals and so obtain the category of
\emph{finite} term graphs; which in the terminology
of~\cite{Hasegawa1997Models}, is the classifying category of the pure sharing
theory over the signature $\Sigma$. Likewise, when $\E = \cat{Set}^\mathbb F$
and $F$ is an endofunctor of the kind considered in
Section~\ref{sec:generalisations}, a sensible choice for $\A$ would be the full
subcategory of $\cat{Set}^\mathbb F$ comprised of the finite coproducts of
representables.

We now show that the categories $\mathscr S[F]$ and $\mathscr S_\infty[F]$
defined above play well with the notion of interpretation described in
Section~\ref{sec:interpretation}. Consider the acyclic case first. Given any
$F$-algebra $c \colon FC \to C$, we obtain maps
\begin{equation}\label{eq:actionmorphisms}
    \mathscr S[F](A, B) \times \E(A, C) \to \E(B, C)
\end{equation}
as follows. Given a cospan $A \xrightarrow f X \xleftarrow g B$ in $\mathscr
S[F]$ and a map $h \colon A \to C$, we may apply Proposition~\ref{prop:lifting}
to obtain an extension $j \colon X \to C$; and composing this with $g$ yields
the required map $jg \colon B \to C$. The point is that this process is
well-behaved with respect to composition of term graphs.
\begin{Prop}
To any $F$-algebra $c \colon FC \to C$, we may associate a functor $\mathscr
S[F] \to \cat{Set}$ given on objects by $A \mapsto \E(A,C)$ and on morphisms
by~\eqref{eq:actionmorphisms}.
\end{Prop}
\begin{proof}
We need only show functoriality; for which we apply the universal properties of
the two constructions given in Proposition~\ref{prop:coalgpushout}
and~\ref{prop:coalgcompose}.
\end{proof}
Transposing these results into the cyclic case we obtain:
\begin{Cor}
To any completely iterative $F$-algebra $c \colon FC \to C$, we may associate a
functor $\mathscr S_\infty[F] \to \cat{Set}$ given on objects by $A \mapsto
\E(A,C)$ and on morphisms by the cyclic analogue of~\eqref{eq:actionmorphisms}.
\end{Cor}
\begin{proof}
By duality.
\end{proof}

Let us conclude by briefly considering the extra structure carried by the
category $\mathscr S[F]$ and $\mathscr S_\infty[F]$; this is very much in the
spirit of~\cite{Hasegawa1997Models} and one can envisage further development
along those lines.
\begin{Prop}
The category $\mathscr S[F]$ admits a symmetric monoidal structure and an
identity-on-objects strict symmetric monoidal embedding $\E^\op \to \mathscr
S[F]$, where $\E^\op$ is equipped with its cartesian monoidal structure.
\end{Prop}
\begin{proof}
The unit of the monoidal structure on $\mathscr S[F]$ is the object $0$ of
$\E$. The tensor product is given on objects by $A \otimes A' = A + A'$ and on
morphisms by $(A \rightarrow X \leftarrow B) \otimes (A' \rightarrow X'
\leftarrow B') = (A + A' \rightarrow X + X' \leftarrow B + B')$; the coalgebra
structure on the left leg of this tensor product being the coproduct in the
category of $L_F$-algebras. The embedding $\E^\op \to \mathscr S[F]$ is the
identity on objects, and on morphisms sends $f \colon A \to B$ to the term
graph
\begin{equation*}\cd{
    B \ar[dr]_{1_B} & & A \ar[dl]^{f} \\ & B}
\end{equation*}
from $B$ to $A$, where $1_B$ is seen as equipped with its unique
$L_F$-coalgebra structure.
\end{proof}


\bibliographystyle{acm}

\bibliography{bibdata}

\end{document}